\documentclass[11pt]{article}
\usepackage{amsmath,amsthm,amssymb,color,verbatim,graphicx,fullpage,url,bm,cleveref}
\usepackage[disable]{todonotes}

\newcommand{\remove}[1]{}
\newtheorem{theorem}{Theorem}[section]
\newtheorem{claim}[theorem]{Claim}
\newtheorem{lemma}[theorem]{Lemma}

\newtheorem{definition}[theorem]{Definition}

\newtheorem{remark}[theorem]{Remark}

\newcommand{\E}{\mathbb{E}}
\newcommand{\F}{\mathbb{F}}

\newcommand{\eps}{\varepsilon}

\newcommand{\x}{\mathbf{x}}
\newcommand{\f}{\mathbf{f}}
\newcommand{\m}{\mathbf{m}}
\renewcommand{\r}{\mathbf{r}}
\newcommand{\g}{\mathbf{g}}
\newcommand{\G}{\mathbf{G}}
\newcommand{\p}{\bm{\pi}}
\newcommand{\y}{\mathbf{y}}
\renewcommand{\v}{\mathbf{v}}

\newcommand{\zz}{\mathbf{z}}

\newcommand{\ftwo}{\mathbb F_2}

\def\E{\mathop{\mathbb{E}}\displaylimits}

\newcommand{\R}[0]{{\ensuremath{\mathbb{R}}}}

\newcommand{\Z}[0]{{\ensuremath{\mathbb{Z}}}}

\title{Optimality of Linear Sketching under Modular Updates}
\author{
Kaave Hosseini\thanks{Supported by NSF grant CCF-1614023.}\\
University of California, San Diego\\
\texttt{skhossei@ucsd.edu}
\and
Shachar Lovett\thanks{Supported by NSF grant CCF-1614023.}\\
University of California, San Diego\\
\texttt{slovett@ucsd.edu}
\and
Grigory Yaroslavtsev\\
Indiana University Bloomington\\
\texttt{grigory.yaroslavtsev@gmail.com}
}

\begin{document}
\maketitle

\begin{abstract}
We study the relation between streaming algorithms and linear sketching algorithms, in the context of binary updates. We show that for inputs in $n$ dimensions,
the existence of efficient streaming algorithms which can process $\Omega(n^2)$ updates implies efficient linear sketching algorithms with comparable cost.
This improves upon the previous work of Li, Nguyen and Woodruff~\cite{LNW14} and Ai, Hu, Li and Woodruff~\cite{AHLW16} which required a triple-exponential number of updates to achieve a similar result for updates over integers. We extend our results to updates modulo $p$ for integers $p \ge 2$, and to approximation instead of exact computation.
\end{abstract}

\section{Introduction}
Linear sketching has emerged in the recent years as a fundamental primitive for algorithm design and analysis including streaming and distributed computing.
Applications of linear sketching include randomized algorithms for numerical linear algebra (see survey~\cite{W14}), graph sparsification (see survey~\cite{M14}), frequency estimation~\cite{AMS99}, dimensionality reduction~\cite{JL84}, various forms of sampling, signal processing, and communication complexity. 
In fact, linear sketching has been shown to achieve optimal space complexity~\cite{LNW14, AHLW16} for processing very long dynamic data streams, which allow elements to be both inserted and deleted. 
Linear sketching is also a frequently used tool in distributed computing --- summaries communicated between processors in massively parallel computational settings are often linear sketches.

In this paper we focus on linear sketches for functions evaluated modulo $p$. Namely, functions of the form $f \colon \mathbb{Z}_p^n \to [0,1]$. Informally, the main result of our work is that for computing such functions linear sketching modulo $p$ achieves almost optimal space complexity in dynamic streaming and distributed simultaneous communication settings. In particular, the setting of
$p$ a power of two (say, 32 or 64) is relevant as CPUs perform computations modulo such powers of two.

\paragraph{Exact sketching for binary data.}
We start with presenting our result in the simplest setting, where $p=2$ and where the output of $f$ is binary. Namely, we are interested in computing a given Boolean function of the form $f(x) \colon \{0,1\}^n \to \{0,1\}$ using only a small sketch of the input.
In this context it is natural to consider sketches which are linear functions over the finite field $\ftwo$. Due to their prominence in design of dynamic streaming graph algorithms and other applications~\cite{AGM12a,AGM12b,CCHM15,K15,BHNT15,MTVV15,EHW16,AKLY16,FT16,AKL17b,CCEHMMV16,KNPWWY17} a study of such $\ftwo$-sketches has been initiated in~\cite{KMSY18}. 

\begin{definition}[Exact $\ftwo$-sketching,~\cite{KMSY18}]\label{def:f2-sketch}
	The \emph{exact randomized $\ftwo$-sketch complexity} with error $\delta$ of a function $f \colon \ftwo^n \to \{0,1\}$ is the smallest integer $k$ such that there exists a distribution 
over linear functions $\ell_1,\ldots,\ell_k: \ftwo^n \to \ftwo$ and a post-processing function $h:\ftwo^k \rightarrow \{0,1\}$ that satisfies:
	$$
	\forall x \in \ftwo^n \colon \Pr_{\ell_1,\ldots,\ell_k,h}\left[
	h(\ell_1(x),\ell_2(x),\ldots, \ell_k(x)) = f(x)\right] \ge 1-\delta.
	$$
\end{definition}

In particular, $\ftwo$-sketches naturally allow one to design algorithms for processing data streams in the XOR update model~\cite{T16} which we refer to as just XOR streams below. 
In this model the input $x \in \{0,1\}^n$ is generated via a sequence of additive updates to its coordinates $i_1, \dots, i_t$ where each $i_j \in [n]$.
Formally, let $x_0 = 0^n$ and let $x_{j} = x_{j - 1} \oplus e_{i_j}$ where $e_{k}$ is the $k$-th unit vector.
This corresponds to flipping the bit in position $i_j$ in $x$ at time $j$ and after applying the sequence of updates the resulting input is $x = x_t$.
The goal of the streaming algorithm is to output $f(x)$.
It is easy to see that by flipping linear functions which depend on $x_{i_j}$ when the update $i_j$ arrives one can maintain an $\ftwo$-sketch through the XOR stream.
Hence the size of the $\ftwo$-sketch gives an upper bound on the space complexity of streaming algorithms in XOR streams. To be more precise,
this equivalence if up to 
logarithmic factors, and follows from a standard application of Nisan's pseudorandom generator~\cite{N90}, which allows to generate the linear functions involved in the sketch in small space. We discuss this further in \Cref{app:prg}.

Whether this simple approach in fact achieves optimal space complexity for streaming applications is one of the central questions in the field.
Two structural results regarding space optimality $\ftwo$-sketching for dynamic streaming are known: 
\begin{enumerate}
\item $\ftwo$-sketches achieve optimal space for streams of length $2^{2^{2^{\Omega(n)}}}$ ~\cite{LNW14,AHLW16,KMSY18}. 
\item $\ftwo$-sketches achieve optimal space for streams of length $\tilde O(n)$ under the assumption that updates are uniformly random~\cite{KMSY18}.
\end{enumerate}
It is open whether optimality of $\ftwo$-sketching holds for short streams without any assumptions about the distribution of updates.
In fact, it was conjectured in~\cite{KMSY18} that such optimality might hold for streams of length only $2n$ (see also Open Problem 78 on \url{http://sublinear.info} from Banff Workshop on Communication Complexity and Applications, 2017).

In this paper we make major progress towards resolving the gap between the two results discussed above. In particular, we show the following theorem.
\begin{theorem}\label{thm:f2}
Let $f \colon \{0,1\}^n \to \{0,1\}$.
Assume that there exists a streaming algorithm for computing $f$ over XOR streams of length $\Omega(n^2)$,
which uses $c$ bits of space. Then the exact randomized $\mathbb{F}_2$-sketch complexity of $f$ is $O(c)$.
\end{theorem}

The proof of \Cref{thm:f2} follows a standard approach in this field, of proving lower bounds for one-way communication protocols. We refer the reader to \Cref{sec:F2} where the model is defined, and to \Cref{theorem:main} which implies \Cref{thm:f2}.

\paragraph{Extensions to updates modulo $p$.}

We now consider the more general streaming model where updates modulo $p$ are allowed, where $p \ge 2$ is an integer. In this model an underlying $n$-dimensional vector $x$ is initialized to $0^n$ and evolves through a sequence of additive updates to its coordinates. These updates are presented to the streaming algorithm as a sequence and have the form $x_i \leftarrow (x_i + \delta_t) \mod p$ changing the $i$-th coordinate by an additive increment $\delta_t$ modulo $p$ in the $t$-th update. Here $\delta_t$ can be an arbitrary positive or negative integer.
In this setting the streaming algorithm is required to output a given function $f$ of $\{0, \dots, p - 1\}^n$ in the end of the stream. 

The definition of the exact randomized $\mathbb{Z}_p$-sketch complexity of $f$ is the natural extension of the definition for $\mathbb{F}_2$.

\begin{definition}[Exact $\mathbb{Z}_p$-sketching]\label{def:Zp-sketch}
	The \emph{exact randomized $\mathbb{Z}_p$-sketch complexity} with error $\delta$ of a function $f \colon \mathbb{Z}_p^n \to \{0,1\}$ is the smallest integer $k$ such that there exists a distribution 
over linear functions $\ell_1,\ldots,\ell_k: \mathbb{Z}_p^n \to \mathbb{Z}_p$ and a post-processing function $h:\mathbb{Z}_p^k \rightarrow \{0,1\}$ that satisfies:
	$$
	\forall x \in \mathbb{Z}_p^n \colon \Pr_{\ell_1,\ldots,\ell_k,h}\left[
	h(\ell_1(x),\ell_2(x),\ldots, \ell_k(x)) = f(x)\right] \ge 1-\delta.
	$$
\end{definition}

\begin{theorem}\label{thm:fp}
Let $f \colon \mathbb{Z}_p^n \to \{0,1\}$.
Assume that there exists a streaming algorithm for computing $f$ over streams 
with modulo $p$ updates of length $\Omega(n^2 \log p)$, which uses $c$ bits of space. Then the exact randomized $\mathbb{Z}_p$-sketch complexity of $f$ is $O(c)$.
\end{theorem}

The proof of \Cref{thm:fp} for prime $p$ is very similar to the proof of \Cref{thm:f2}. However, for non-prime $p$ the proof is a bit more involved. We define the relevant model in \Cref{sec:groups},
and the relevant theorem that implies \Cref{thm:fp} is \Cref{theorem:main3}.

\paragraph{Extensions to approximation.}

It is also natural to consider real-valued functions $f \colon \{0, 1\}^n \to [0,1]$ and to allow the streaming algorithm to compute $f$ with error $\epsilon$ (see e.g.~\cite{YZ18}). It turns out that technically, a convenient notion of approximation is $\ell_2$ approximation. Namely, a randomized function $\g$ (computed by a streaming protocol, or a sketching protocol) $\eps$-approximates $f$ if
$$
\E \left[|f(x) - \g(x)|^2 \right] \le \eps \qquad \forall x \in \{0,1\}^n.
$$
A similar definition holds for more general functions $f:\mathbb{Z}_p^n \to [0,1]$.
The definition of approximate randomized $\mathbb{Z}_p$-sketch complexity is the natural extension of the previous definitions.

\begin{definition}[Approximate $\mathbb{Z}_p$-sketching]\label{def:Zp-sketch}
	The \emph{approximate randomized $\mathbb{Z}_p$-sketch complexity} with error $\delta$ of a function $f \colon \mathbb{Z}_p^n \to [0,1]$ is the smallest integer $k$ such that there exists a distribution 
over linear functions $\ell_1,\ldots,\ell_k: \mathbb{Z}_p^n \to \mathbb{Z}_p$ and a post-processing function $h:\mathbb{Z}_p^k \rightarrow \{0,1\}$ that satisfies:
	$$
	\forall x \in \mathbb{Z}_p^n \colon \E_{\ell_1,\ldots,\ell_k,h}\left[
	\left| h(\ell_1(x),\ell_2(x),\ldots, \ell_k(x)) - f(x) \right|^2\right] \le \delta.
	$$
\end{definition}

\begin{theorem}\label{thm:fp-approx}
Let $f \colon \mathbb{Z}_p^n \to \{0,1\}$.
Assume that there exists a streaming algorithm which $\eps$-approximates 
$f$ over streams with modulo $p$ updates of length $\Omega(n^2 \log p)$, which uses $c$ bits of space. Then the approximate randomized $\mathbb{Z}_p$-sketch complexity of $f$ with error $O(\eps)$ is $O(c)$.
\end{theorem}

We develop the machinery needed to handle approximation in two steps. First, in \Cref{sec:F2_approx} we prove it for $p=2$, see in particular \Cref{theorem:main2}.
For general $p$ it is done in \Cref{sec:groups}, 
where the relevant theorem which implies
\Cref{thm:fp-approx} is  \Cref{theorem:main4}. 

\paragraph{Techniques.}

The result in \Cref{thm:f2} starts with a standard connection between streaming algorithms and a multi-party one-way communication game.
In this game there are $N$ players each holding an input in $\{0,1\}^n$.
These inputs are denoted as $x_1, \dots, x_N$, respectively.
The players communicate sequentially in $N$ rounds where in the $i$-th round the $i$-th player sends a message to the $(i + 1)$-th player.
The players have access to shared randomness and the $i$-th message can depend on the $(i-1)$-th message, the input $x_i$ and the shared randomness. 
The message sent in the final $N$-th round should be equal to $f(x_1 + \dots + x_N)$.
In fact, our proof works in a more general model where the $i$-th message can depend on all messages sent by previous players.
We refer to the above model as the \textit{one-way broadcasting communication model}.

Our technical result which 
implies \Cref{thm:f2}, \Cref{theorem:main}, shows that in any protocol which allows to compute $f$ at least one of the messages sent by the players has to be of size $\Omega(k)$ where $k$ is the smallest dimension of an $\ftwo$-sketch for $f$.
This immediately implies a space lower bound of $\Omega(k)$ for streaming algorithms in the XOR update model.
Indeed, if a streaming algorithm with smaller space existed then the players could just pass its state as their message after applying updates corresponding to their local inputs.

Our proof of the communication lower bound proceeds as follows. First, it will be easier to present the argument for $N+1$ players instead of $N$ players.
Assume that there exists a communication protocol which succeeds with probability $q$ and sends at most $c$ bits in every round. This protocol has to succeed for any distribution of the inputs. So, fix a ``hard" distribution $D$ over inputs $x \in \{0,1\}^n$. 

We sample the inputs $x_1,\ldots,x_{N+1} \in \{0,1\}^n$
to the $N$ players as follows: first, sample $x \sim D$. Then, sample $x_1,\ldots,x_N \in \{0,1\}^n$ uniformly. Finally, set $x_{N+1} = x + x_1 + \ldots + x_N$, so that the sum (modulo two) of the inputs to the $N+1$ players equals $x$. 

Next, an averaging argument then shows that there is a transcript (sequence of messages) $\pi=(m_1,\ldots,m_{N})$ of the first $N$ players such that:
\begin{enumerate}
\item[(i)] Conditioned on the transcript $\pi$, the protocol computes $f$ correctly with probability $\approx q$.
\item[(ii)] The probability for the 
transcript $\pi$ is not too tiny, concretely $\approx 2^{-cN}$.
\end{enumerate}

Once we fixed the transcript $\pi$, note that the output of the protocol depends only on the last input $x_{N+1}$, which we denote by $F(x_{N+1})$. Recall that by our construction, $x_{N+1} = x + x_1+\ldots+x_N$. Define sets $A_1,\ldots,A_N \subset \{0,1\}^n$ such that whenever $x_1 \in A_1,\ldots,x_N \in A_N$, the players send the transcript $\pi$. By (ii) above it holds that the density of a typical $A_i$ is approximately $2^{-c}$. Then, if we sample $y_i \in A_i$ uniformly then by (i) we have
$$
\Pr_{x \sim D, y_i \in A_i} \left[ F(x + y_1 + \ldots + y_N) = f(x) \right] \approx q.
$$

The next step is to apply Fourier analysis. In particular, we rely of Chang's lemma~\cite{chang2002polynomial}. 
This allows us to deduce that there exists a subspace $V \subset \ftwo^n$ of co-dimension $O(c)$ such that, if we sample in addition $v \in V$ uniformly, then
$$
\Pr_{x \sim D, y_i \in A_i, v \in V} \left[ F(x + y_1 + \ldots + y_N + v) = f(x) \right] \approx q.
$$
Concretely, $V$ is chosen to be orthogonal to the common large Fourier coefficients of the indicator functions of $A_1,\ldots,A_N$. In order for this to hold, it is necessary to choose $N$ large enough so that the sum $y_1 + \ldots + y_N$ ``mixes" enough in the group $\mathbb{F}_2^n$. It turns out that $N=\Omega(n)$ is sufficient for this.

This allows us to define a randomized $\mathbb{F}_2$-sketching protocol. Consider the quantity 
$$
g(x) = \Pr_{y_i \in A_i, v \in V} \left[ F(x + y_1 + \ldots + y_N + v) = 1 \right].
$$
The function $g(x)$ depends only on the coset $x+V$, and hence can be computed by a randomized $\mathbb{F}_2$-sketching protocol with complexity equals to the co-dimension of $V$.

The results for updates modulo $p$ (\Cref{thm:fp}) and approximation (\Cref{thm:fp-approx}) follow the same general scheme, except that now players are holding inputs in $\mathbb Z_p^n$ and we convert any $c$-bit protocol with small  error into a sketch modulo $p$ of dimension $O(c)$ which has a similar error. 
In order to achieve mixing in this setting the required number of players is $N=\Omega(n \log p)$.

\paragraph{Distributed computing in the simultaneous communication model.}
Our results imply that lower bounds on linear sketches modulo $p$ immediately lead to lower bounds for computing additive functions in the simultaneous communication complexity (SMP) model. In this model~\cite{BK97,BGKL03} there are $N$ players and a coordinator, who are all aware of a function $f \colon \mathbb{Z}_p^n \to [0,1]$. 
The players have inputs $x_1, \dots, x_N \in \mathbb{Z}_p^n$ and must send messages of minimal size to the coordinator so that the coordinator can compute $f(x_1, \dots, x_N)$ using shared randomness.
If $f$ is additive, i.e. of the form $f(x_1 + \dots + x_N)$ then this is strictly harder than the one-way broadcasting model described above.
Note that dimension of the best linear sketch modulo $p$ for $f$ still translates to a protocol for the SMP model.


\paragraph{Previous work.}
Most closely related to ours are results of~\cite{LNW14} and~\cite{AHLW16} which stemmed from the work of~\cite{G08}.
In particular~\cite{LNW14} shows that under various assumptions about the updates turnstile streaming algorithms can be turned into linear sketches over integers with only a $O(\log p)$ multiplicative loss in space.
While this is similar to our results, these approaches inherently require extremely long streams of adversarial updates (of length triply exponential in $n$ in~\cite{LNW14}) as they essentially aim to fail any small space algorithm (modeled as a finite state automaton) using a certain sequence of updates.  
Furthermore, the results of~\cite{LNW14} rely on a certain ``box constraint'' requirement.
This requirement says that correctness of the streaming algorithm should be guaranteed for the resulting input $x \in \{-m, \dots, m\}^n$ even if the intermediate values of the coordinates of $x$ throughout the stream are allowed to be much larger than $m$ in absolute value. While this requirement has been subsequently removed in~\cite{AHLW16}, their results again impose a certain constraint on the class of streaming algorithms they are applicable to.
In particular, their Theorem 3.4 which removes the ``box constraint'' is only applicable to algorithms which use space at most $\frac{c \log m}{n}$ which is only non-trivial for $m = \Omega(2^n)$.

It has been open since the work of~\cite{LNW14} and~\cite{AHLW16} whether similar results can be obtained using the tools from communication complexity, as has been the case for most other streaming lower bounds.
While our results don't apply directly to updates over integers, a key component of the~\cite{LNW14,AHLW16} technique is to first reduce general automata to linear sketching modulo fixed integers. Hence our result can be seen as an alternative to their reduction which is specific to modular updates and is obtained through communication complexity tools.

Another related line of work is on communication protocols for XOR-functions~\cite{SZ08,MO09,TWXZ13,HHL16,Y16,KMSY18,YZ18}.
For inputs $x_1, \dots, x_N \in \mathbb F_2^n$ a multi-parity XOR-function is defined as $f(x_1 + \dots + x_N)$.
For the case of $p = 2$ our results are using one-way broadcasting communication complexity for the corresponding XOR-function of interest.
While the communication complexity of XOR-functions has been studied extensively in the two-party communication model, to the best of our knowledge prior to our work it hasn't been considered in the one-way multi-party setting.

\section{Sketching for  $f:\ftwo^n\rightarrow \{0,1\}$}
\label{sec:F2}

\subsection{Model}
We use regular letters $x,f$ for deterministic objects and bold letters $\x,\f$ for random variables. 

\paragraph{Streaming protocol.}
Let $F:(\ftwo^n)^N \to \{0,1\}$ be an $N$-player function, where the players' inputs are $x_1,\ldots,x_N \in \ftwo^n$. We assume that the players have access to shared randomness $\r \in \{0,1\}^r$.

A \emph{streaming protocol} for $F$ with $c$ bits of communication is defined as follows. The players send messages in order, where the $i$-th player's message $\m_i$ depends on her input $x_i$, the previous player's message $\m_{i-1}$ and the shared randomness $\r$. That is, 
\begin{align*}
&\m_1 = M_1(x_i, \r),\\
&\m_i = M_i(x_i, \m_{i-1}, \r), \qquad i=2,\ldots,N 
\end{align*}
where $M_i:\ftwo^n \times \{0,1\}^c \times \{0,1\}^r \to \{0,1\}^c$ for $1 \le i\leq N-1$ and 
$M_N:\ftwo^n \times \{0,1\}^c \times \{0,1\}^r \to \{0,1\}$, where the output of the protocol is the last message sent $\m_N \in \{0,1\}$. We may write it as $\m_N = G(x_1,\ldots,x_N,\r)$, where $G$ respects the protocol structure:
$$
G(x_1,\ldots,x_N,\r) = M_n(\ldots,M_2(x_2, M_1(x_1, \r), \r),\r).
$$
The protocol computes $F$ correctly with probability $q$, if for all possible inputs $x_1,\ldots,x_N \in \ftwo^n$, it holds that
$$
\Pr \left[G(x_1,\ldots,x_N,\r) = F(x_1,\ldots,x_N)\right] \ge q.
$$

\paragraph{One-way broadcasting.}
A one-way broadcasting protocol is a generalization of a streaming protocol. We introduce this model as our simulation theorem extends to this model seamlessly.

In this model, the message sent by the $i$-th player is seen by all the players coming after her. Equivalently, the $i$-th player's message may depend on $\m_1,\ldots,\m_{i-1}$, 
$$
\m_i = M_i(x_i, \m_1,\ldots,\m_{i-1}, \r).
$$
The notion of a protocol computing $F$ correctly with probability $q$ is defined analogously.

\paragraph{Linear sketches.}
A function $g:\ftwo^n \to \{0,1\}$ is a $k$-linear-junta if $g(x) = h(\ell_1(x),\ldots,\ell_k(x))$, where each $\ell_i:\ftwo^n \to \ftwo$ is a linear function,
and $h:\ftwo^k \to \{0,1\}$ is an arbitrary function. A linear sketch of cost $k$ is a distribution $\g:\ftwo^n \to \{0,1\}$ over $k$-linear-juntas. It computes $f$
with success probability $q$ if, for every input $x \in \ftwo^n$, it holds that
$$
\Pr\left[\g(x) = f(x)\right] \ge q.
$$

\paragraph{Simulation theorem.} Let $f:\ftwo^n \to \{0,1\}$ and $F:(\ftwo^n)^N \to \{0,1\}$ be the $N$-player function defined by
$$
F(x_1,\ldots,x_N) = f(x_1 + \ldots + x_N).
$$ 
We show that if there are sufficiently many players ($N$ is large enough) and $F$ has an efficient one-way broadcasting protocol, then $f$ also has an efficient linear sketch.

\begin{theorem}\label{theorem:main}
Let $N \ge 10n$ and suppose
 that $F$ has a one-way broadcasting protocol with $c$ bits of communication per message and success probability $q$.
Then there exists a linear sketch of cost $k$ which computes $f$ with success probability $q-2^{-\Omega(N)}$, where $k=O(c)$.
\end{theorem}
\begin{remark}
We remark that the statement and proof of \Cref{theorem:main} generalizes straightforwardly to the case that $f:\F_p^n\rightarrow \{0,1\}$ for any prime $p$; the only difference is that we have to ensure that $N\geq 10n\log p$. We provide the proof over $\ftwo$ to slightly simplify the notation.  
\end{remark}

\subsection{Proof of \Cref{theorem:main}}
By Yao's minimax principle, it will suffice to show that 
for any distribution $D$ over $\ftwo^n$, 
there exists a $k$-linear-junta $g:\ftwo^n \to \{0,1\}$ such that
$$
\Pr_{\x \sim D}[g(\x) = f(\x)] \ge q-2^{-\Omega(N)}.
$$

Fix a distribution $D$.
We consider the following distribution over the inputs. It will be easier to assume we have $N+1$ players instead of $N$ players.
First, sample $\x_1,\ldots,\x_{N} \in \ftwo^n$ uniformly, and let $\x \sim D$. Set $\x_{N+1} = \x_1 + \ldots + \x_{N} + \x$.
Under this input distribution, there exists a fixed choice of the shared randomness which attains success probability $\ge q$. Namely, 
there is a fixed $r^*$ such that for $\m_1=M_1(\x_1,r^*), \m_2=M_2(\x_2,\m_1, r^*)$, etc, it holds that
\begin{equation}
\label{eq:prob}
\Pr_{\x_1,\ldots,\x_{N+1}} \left[G(\m_1,\ldots,\m_{N+1},r^*) = f(\x)\right] \ge q.
\end{equation}

Let $\p=(\m_1,\ldots,\m_{N}) \in \{0,1\}^{cN}$ denote the messages of the first $N$ players. 
For every possible value $\pi$ of $\p$, define 
$$
a(\pi) = \Pr[\p=\pi], \qquad b(\pi) = \Pr \left[G(\m_1,\ldots,\m_{N+1}, r^*) = f(\x) \; | \; \p=\pi\right].
$$
Then we may rewrite \Cref{eq:prob} as
\begin{equation*}
\label{eq:prob_ab}
\sum_{\pi} a(\pi) b(\pi) \ge q.
\end{equation*}
Let $\delta=2^{-N}$. By averaging, there exists a choice of $\pi=(m_1,\ldots,m_{N})$ such that
\begin{enumerate}
\item[(i)] $a(\pi) \ge \delta 2^{-cN} = 2^{-(c+1)N}$.
\item[(ii)] $b(\pi) \ge q-\delta$.
\end{enumerate}
Define sets $A_i=\{x_i \in \ftwo^n: M_i(x_i,m_1,\ldots,m_{i-1}, r^*)=m_i\}$ for $i=1,\ldots,N$ so that
$$
\left[ \p = \pi \right] \quad \Leftrightarrow \quad \left[ \x_1 \in A_1,\ldots,\x_{N} \in A_{N} \right].
$$
Let $\alpha_i = \frac{|A_i|}{2^n}$ denote the density of $A_i$. Condition (i) translates to
\begin{equation}
\label{eq:a}
\prod_{i=1}^{N} \alpha_i = a(\pi) \ge 2^{-(c+1)N}.
\end{equation}

Next, conditioned on $\p=\pi$, the one-way broadcasting protocol simplifies. First, define $h:\ftwo^n \to \{0,1\}$ as
$$
h(x_{N+1}) = G(m_1,\ldots,m_{N}, M_{N+1}(x_{N+1},m_1,\ldots,m_N,r^*)).
$$
Let $\y_1 \in A_1,\ldots,\y_N \in A_N$ be uniformly and independently chosen. Condition (ii) translates to
\begin{equation}
\label{eq:h}
\Pr\left[h(\x+\y_1+\ldots+\y_N) = f(\x)\right] = b(\pi) \ge q-\delta.
\end{equation}

We next apply Fourier analysis. Let us quickly set some common notation in the following paragraph.
\paragraph{Fourier analysis.}
Let $f:\ftwo^n\rightarrow \R $ be  be a function. Then given $\gamma\in \ftwo^n$, the Fourier coefficient $\widehat{f}(\gamma)$ is define as $\widehat{f}(\gamma) = \E_{x\in \ftwo^n} f(x) (-1)^{\langle x,\gamma\rangle}$. The function $f$ can be expressed in the Fourier basis as $f(x) = \sum_{\gamma\in \ftwo^n} \widehat{f}(\gamma) (-1)^{\langle x,\gamma\rangle}$. Given two functions $f,g:\ftwo^n\rightarrow \R$, their convolution  $f*g:\ftwo^n\rightarrow \R$, is defined by $f*g(x) = \E_{y\in \ftwo^n} f(y)g(x+y)$. Moreover, we have the equality $\widehat {f*g}(\gamma) =  \widehat{f}(\gamma) \widehat{g}(\gamma)$ for all $\gamma \in \ftwo^n$. Given a set $A \subset \ftwo^n$ of density $\alpha=\frac{|A|}{2^n}$, define its normalized indicator function $\varphi_A:\ftwo^n \to \R$ as
$$
\varphi_A(x) = 
\begin{cases}
1/\alpha & \text{if } x \in A\\
0 & \text{otherwise}
\end{cases}\;.
$$
Note that under this normalization, $\widehat{\varphi_A}(0) = \E[\varphi_A]=1$ and $|\widehat{\varphi_A}(\gamma)| \le 1$ for all $\gamma \in \F_p^n$.

Going back to the proof, for technical reasons we switch to the $\{-1,1\}$  instead of the $\{0,1\}$. Define the functions $h' = (-1)^h$ and $f' = (-1)^f$. We work with $h',f'$ instead. Note that \Cref{eq:h} translates to  
\begin{equation}
\label{eq:hh}
\Pr\left[h'(\x+\y_1+\ldots+\y_N)  f'(\x)=1\right] = b(\pi) \ge q-\delta.
\end{equation}
 We use the following immediate consequence of Chang's lemma~\cite{chang2002polynomial}.
\begin{lemma}
\label{lemma:chang}
Let $A \subset \ftwo^n$ of density $\alpha = \frac{|A|}{2^n}$. Let $\gamma_1,\ldots,\gamma_k \in \ftwo^n$ be linearly independent. Then
$$
\sum_{i=1}^k |\widehat{\varphi_A}(\gamma_i)|^2 \le 8\log 1/\alpha.
$$
\end{lemma}
Let $S \subset \ftwo^n$ be a set of ``noticeable" Fourier coefficients of $A_1,\ldots,A_N$, defined as follows. 
First, define 
$$
B=\left \{i \in [N]: \frac{|A_i|}{2^n} \ge 2^{-2(c+1)} \right \}.
$$
\Cref{eq:a} implies that $|B| \ge N/2$. Then, define 
$$
S =\left\{\gamma \in \ftwo^n: \sum_{i \in B} |\varphi_{A_i}(\gamma)|^2 \ge |B|/2\right\}.
$$
Let $\gamma_1,\ldots,\gamma_k$ be a maximal set of linearly independent elements in $S$. \Cref{lemma:chang} implies that $k=O(c)$. Thus, there exists a subspace $U \subset \ftwo^n$ of dimension $k$ such that $S \subset U$.
Let $V \subset \ftwo^n$ be the orthogonal subspace to $U$. 

\begin{claim}\label{claim:fourier}
Let $\y_1 \in A_1,\ldots, \y_N \in A_N, \v \in V$ be chosen uniformly and independently. Then for every $x \in \ftwo^n$ it holds that
$$
\left| \E[h'(x+\y_1+\ldots+\y_N)] - \E[h'(x+\y_1+\ldots+\y_N+\v)] \right| \le 2^{n}2^{-N/8}.
$$
\end{claim}

\begin{proof}
We rewrite both expressions using their Fourier expansion:
\begin{align*}
\E[h'(x + \y_1 + \dots \y_N)] = \varphi_{A_1} * \dots * \varphi_{A_N} * h' (x) = \sum_{\gamma \in \ftwo^n} (-1)^{\langle \gamma, x \rangle} \widehat {h'}(\gamma) \prod_{i = 1}^N \widehat{\varphi_{A_i}}(\gamma)
\end{align*}
and similarly
\begin{align*}
\E[h'(x + \y_1 + \dots \y_N + \v)] = \sum_{\gamma \in \ftwo^n}(-1)^{\langle \gamma, x \rangle} \widehat {h'}(\gamma) \widehat{\varphi_V}(\gamma) \prod_{i = 1}^N \widehat{\varphi_{A_i}}(\gamma).
\end{align*}
As $V$ is a subspace, we have $\widehat{\varphi_V}(\gamma) = 1_{U}(\gamma)$. We can thus bound
$$
\left| \E[h'(x+\y_1+\ldots+\y_N)] - \E[h'(x+\y_1+\ldots+\y_N+\v)] \right| \le \sum_{\gamma \notin U} \prod_{i=1}^N \left| \widehat{\varphi_{A_i}}(\gamma) \right|.
$$
If $\gamma \in S$ then also $\gamma \in U$. Otherwise, by the construction of $S$, $|\widehat{\varphi_{A_i}}(\gamma)|^2 \le 1/2$ for at least $|B|/2$ elements $i \in [N]$. We thus have
$$
\left| \E[h'(x+\y_1+\ldots+\y_N)] - \E[h'(x+\y_1+\ldots+\y_N+\v)] \right| \le 2^{n}2^{- |B|/4} \le 2^{n} 2^{- N/8}.
$$
\end{proof}
 Moreover, using the assumption $N \ge 10n $ provides
$$
\left| \E[h'(x+\y_1+\ldots+\y_N)] - \E[h'(x+\y_1+\ldots+\y_N+\v)] \right| \le 2^{-\Omega(N)}.
$$
Now, by \Cref{eq:hh} we already have 
$$
 \E[h'(\x+\y_1+\ldots+\y_N)f'(\x)]\geq 2q-1-2^{-\Omega(N)}
$$
allowing us to obtain
$$
 \E[h'(\x+\y_1+\ldots+\y_N+\v)f'(\x)]\geq 2q-1-2^{-\Omega(N)}
$$
implying 
$$
\Pr \left[h(\x+\y_1+\ldots+\y_N+\v) = f(\x)\right] \geq q-2^{-\Omega(N)}
$$
To conclude the proof, define the function $w:\ftwo^n \to [0,1]$ as
$$
w(x) = \E\left[ h(x+\y_1+\ldots+\y_N+\v) \right].
$$
Note that $w(x) = W(\ell_1(x),\ldots,\ell_k(x))$, where $\ell_1,\ldots,\ell_k$ are a basis for $U$, and $W:\ftwo^k \to [0,1]$. 
Define a randomized function $\G:\ftwo^k \to \{0,1\}$, where $\Pr[\G(z)=1]=W(z)$ independently for each $z \in \ftwo^k$.
Define $\g(x) = \G(\ell_1(x),\ldots,\ell_k(x))$, which is a randomized $k$-linear-junta and observe that $\g(x)$ and $h(x+\y_1+\cdots+\y_N+\v)$ have the same distribution for every $x\in \ftwo^n$.
Therefore, we have that
\begin{align*}
\Pr \left[\g(\x) = f(\x) \right] = \Pr \left[h(\x+\y_1+\ldots+\y_N+\v) = f(\x)\right] \geq q - 2^{-\Omega(N)},
\end{align*}

Finally, note that by an averaging argument, we can fix the internal randomness of $\g$ to obtain a $k$-linear-junta $g$ so that
$$
\Pr_{\x \sim D} \left[g(\x) = f(\x) \right] \ge q - 2^{-\Omega(N)}.
$$
This concludes the proof.

\section{Sketching for  $f:\ftwo^n\rightarrow [0,1]$}
\label{sec:F2_approx}
In this section, we approximate a given function $f:\ftwo^n\rightarrow [0,1]$ with an additive error. Both the model and the proof are similar to the previous case.
\subsection{Model} 
\paragraph{Protocols.}
The notions of \emph{streaming protocol} and \emph{one-way broadcasting protocol} with $c$ bits of communication are defined similar as before, where the only difference is that the last message $\m_N$ takes values in $[0,1]$ instead of $\{0,1\}$.

The protocol $G$ is said to \emph{$\eps$-approximate}  $F:(\ftwo^n)^N \to [0,1]$ if for all possible inputs $x_1,\ldots,x_N \in \ftwo^n$, it holds that
$$
\E_{\r} \left[\left|G(x_1,\ldots,x_N,\r) - F(x_1,\ldots,x_N)\right|^2\right] \leq \eps .
$$

\paragraph{Linear sketches.}
 A linear sketch of cost $k$ is a distribution $\g:\ftwo^n \to [0,1]$ over $k$-linear-juntas, where a $k$-linear-junta $g:\ftwo^n \to [0,1]$ is defined as before. The linear sketch $\g$ \emph{$\eps$-approximates} $f:\ftwo^n \to [0,1]$
if, for every  $x \in \ftwo^n$, it holds that
$$
\E\left[\left|\g(x) - f(x) \right|^2\right] \leq \eps.
$$

We prove the following theorem in the rest of the section.
\begin{theorem}\label{theorem:main2}
Let $f:\ftwo^n \to [0,1]$ and assume $N \ge 10n  $ and $F:(\ftwo^n)^N \to [0,1]$  is defined by
$
F(x_1,\ldots,x_N) = f(x_1 + \ldots + x_N)
$. Suppose that $F$ is $\eps$-approximated by a one-way broadcasting protocol with $c$ bits of communication per message. Then there is a linear sketch $\g:\ftwo^n\rightarrow[0,1]$ of cost $k$ that $\eps'$-approximates $f$, where $k=O(c)$ and $\eps' = 2\eps+2^{-\Omega(N)}$. 
\end{theorem}
\begin{remark}
In this case as well, the proof directly generalizes to the case that $f:\F_p^n\rightarrow [0,1]$ conditioned on $N\geq 10n\log p$.
\end{remark}
\subsection{Proof of \Cref{theorem:main2}}
The proof is similar to the proof of \Cref{theorem:main}. We point out the necessary modifications. Same as before, we fix an arbitrary  distribution $D$ over $\ftwo^n$, 
and show that there exists a $k$-linear-junta $g:\ftwo^n \to [0,1]$ such that
$$
\E_{\x\sim D}\left[\left|g(\x) - f(\x) \right|^2\right] \leq \eps'.
$$
To do so, obtain a function $h:\ftwo^n\rightarrow [0,1]$ and sets $A_1,\cdots,A_N$ as before, so that for $\x\sim D, \y_i\in A_i$, we have
\begin{equation}
\label{eq:happrox}
\E \left[\left|h(\x+\y_1+\ldots+\y_N) - f(\x)\right|^2\right]\leq \eps +2^{-N}.
\end{equation}
Now we switch to the exponential basis. Define functions $f',h':\ftwo^n\rightarrow \mathbb{C}$ by 
$$f'(x) = e^{ i f(x)}, h'(x) = e^{- i h(x)}. $$
where $e= 2.71828\cdots$ is Euler's constant. Note that if $f(x) = h(x)$, then $\operatorname{Re}\left[f'(x)h'(x)\right]=1$, where $\operatorname{Re}[z]$ is the real component of $z$.  We need the following claim.
\begin{claim}\label{claim:conversion}
Let $\zz$ be a $[-1,1]$-valued random variable. Then, 
$$1-\frac{1}{2}\E\left[\zz^2\right]\leq\E \left[\operatorname{Re}[e^{i \zz}]\right]\leq 1-\frac{1}{3}\E\left[\zz^2\right].
$$  
\end{claim}
\begin{proof}
The Taylor expansion of $\operatorname{Re}e^{iz}=cos(z)$ is $1-\frac{z^2}{2!}+\frac{z^4}{4!}-\cdots$. Therefore
$$
1-\frac{z^2}{2} \leq \operatorname{Re}[e^{i z}]\leq 1-\frac{z^2}{3} 
$$
provided that $z\in [-1,1]$. 
\end{proof}
Using the lower bound in \Cref{claim:conversion}, by taking  $\zz = h(\x+\y_1+\ldots+\y_N) - f(\x)$, we get 
\begin{equation*}
\E\left[\operatorname{Re}[h'(\x+\y_1+\ldots+\y_N)  f'(\x)]\right]\geq 1-\eps/2 - 2^{-\Omega(N)}
\end{equation*}
We apply \cref{claim:fourier} same as before to obtain subspaces $V,U$ and 
\begin{equation}\label{eq:hmultapproxv}
\E\left[\operatorname{Re}[h'(\x-\y_1-\ldots-\y_N+\v)  f'(\x)]\right]\geq 1-\eps/2 - 2^{-\Omega(N)}
\end{equation}
Define a randomized $k$-linear-junta $\r:\F_p^n\rightarrow [0,1]$ as follows. Sample $\y_1,\cdots,\y_N,\v$. Then for every $u\in U$ and $v\in V+u$, set
$$\r(v) = h(u-\y_1-\cdots-\y_N+\v).$$
Observe that for every $x\in \F_p^n$, the randomized functions 
$\r(x)$ and $ h(x-\y_1-\ldots-\y_N+\v)$ have identical distributions, and therefore, $ e^{-\pi i \r(x) }$ has the same distribution as $h'(x-\y_1-\cdots-\y_N+\v)$. Combining this with \cref{eq:hmultapproxv} implies
\begin{equation*}
\E\left[\operatorname{Re}[e^{-\pi i \r(\x) }  f'(\x)]\right]\geq 1-\eps/2 - 2^{-\Omega(N)}
\end{equation*}
By an averaging argument, there is a $k$-linear-junta $g:\F_p^n\rightarrow [0,1]$ that 
\begin{equation*}
\E\left[\operatorname{Re}[ e^{-\pi i g(\x)} f'(\x)]\right]\geq 1-\eps/2 - 2^{-\Omega(N)}.
\end{equation*}
Finally, by using the upper bound in \cref{claim:conversion}, we get that 
\begin{equation*}
 \E\left[|g(\x)-f(\x)|^2\right] \leq 1 - 4\E\left[\operatorname{Re}[ e^{-\pi i g(\x)} f'(\x)]\right]\leq 2\eps + 2^{-\Omega(N)}
\end{equation*}
which finishes the proof.

\section{Sketching over abelian groups of bounded exponent}
\label{sec:groups}
Let $G$ be a finite abelian group. We generalize \Cref{theorem:main} and \Cref{theorem:main2} to the case where $f:G\rightarrow \{0,1\}$ and $f:G\rightarrow [0,1]$, respectively. Even though we are mostly interested in $G=\mathbb{Z}_p^n$, it turns out to be useful to consider this more general formulation.
We introduce the required modifications to the definitions and the proofs. 
\paragraph{Protocols.}
The concept of broadcasting protocol and streaming protocol for the function $F(x_1,\cdots,x_N) = f(x_1+\cdots+x_N)$ is defined as before.

\paragraph{Sketching.}
We modify the previous definition of sketching so that it will be meaningful for arbitrary abelian groups.
Let $H$ be an arbitrary subgroup of $G$. Let $Q = G/H$ be the quotient group. A function $g:G\rightarrow \{0,1\}$ is  \emph{$H$-invariant} if it is constant on every coset $H+w$ for all $w\in Q$. Note that such $g$ can be factored as $g(x) = h(q(x))$ where $q:H\rightarrow Q$ is defined by $q(x) = x+H$ and $h: Q\rightarrow \{0,1\}$ is an arbitrary function. A function $g:G\rightarrow \{0,1\}$ has \emph{linear complexity} $r$ if there is a subgroup $H\leq G$ so that $g$ is $H$-invariant and also $|G/H|\leq r$. Note that for functions $g:\ftwo^n\rightarrow \{0,1\}$, the notion of  $k$-linear-junta is equivalent to linear complexity $2^k$. 

A \emph{linear sketch} of \emph{complexity $r$} is a distribution $\g:G\rightarrow \{0,1\}$ over functions $g:G\rightarrow \{0,1\}$ of linear complexity $r$.
We have the following two theorems for functions $g:G\rightarrow \{0,1\}$
and $g:G\rightarrow [0,1]$.

\paragraph{Simulation theorems.}
Before stating the theorems, we need to introduce one parameter of the group $G$ that will be important here. 
Let $G \cong \Z_{m_1}\times \cdots \times\Z_{m_n}$ where each $\Z_{m_i}$ is the cyclic group of order $m_i$.  Let $m$ --- called the \emph{exponent} of $G$ --- be the least common multiple of $m_i$'s. Now, we can state the simulation theorems. 
Same as before, let $f:G \rightarrow \{0,1\}$ (respectively, $f:G \rightarrow [0,1]$) and define $F: G^N\rightarrow \{0,1\}$ (respectively, $F: G^N\rightarrow [0,1]$) by $F(x_1,\cdots,x_N) = f(x_1+\cdots+x_N)$.
\begin{theorem}\label{theorem:main3}
Let $G \cong  \Z_{m_1}\times \cdots \times  \Z_{m_n}$ be an abelian group of exponent $m$. Let $N \ge 10n\log m$ and suppose that $F$ has a one-way broadcasting protocol with $c$ bits of communication per message and success probability $q$.
Then there exists a linear sketch of complexity $r$ which computes $f$ with success probability $q-2^{-\Omega(N)}$, where $r=m^{O(c)}$.
\end{theorem}
And similarly, for bounded real-valued functions we have the following.
\begin{theorem}\label{theorem:main4}
Let $G \cong  \Z_{m_1}\times \cdots \times  \Z_{m_n}$ be an abelian group of exponent $m$ and $N\geq 10n\log m$.
 Suppose that $F$ is $\eps$-approximated by a one-way broadcasting protocol with $c$ bits of communication per message. Then there is a linear sketch $\g:\ftwo^n\rightarrow[0,1]$ of complexity $r$ that $\eps'$-approximates $f$, where $r=m^{O(c)}$ and $\eps' = 2\eps+2^{-\Omega(N)}$. 
\end{theorem}

We need to provide suitable versions of \Cref{lemma:chang} and \Cref{claim:fourier} here. The other parts of the proof are same as before.
We first have to introduce some notation to do Fourier analysis. 
\paragraph{Fourier analysis.}
 A character $\gamma:G\rightarrow \mathbb{C}^*$ of $G$ is a homomorphism to the group $\mathbb{C}^*$. That is, for every $x,y \in G$ we have $\gamma(x+y) = \gamma(x)\gamma(y)$ and $\gamma(0) = 1$. The dual group of $G$ denoted by $\widehat{G}$ is the group of all characters of $G$. $\widehat{G}$ has the group structure introduced by $(\gamma_1+ \gamma_2)(x) =\gamma_1(x)\gamma_2(x)$. In fact $\widehat{G}$ is isomorphic to $G$. Given any function $f:G\rightarrow \mathbb{C}$, we can write $f$ in its Fourier basis as 
$$
f(x) = \sum_{\gamma \in \widehat{G}} \widehat{f}(\gamma)\gamma(x) 
$$
where $$\widehat{f}(\gamma) = \E_{x\in G} f(x)\overline{\gamma(x)}$$
and $\overline{\gamma(x)}$ is the complex conjugate of $\gamma(x)$. Moreover the convolution operator is defined as before. Given $f,g:G\rightarrow \mathbb{C}$, write $f*g(x) = \E_{y\in G} f(y)g(x-y)$ for $x\in G$, which leads to $\widehat{f*g}(\gamma) = \widehat{f}(\gamma)\widehat{g}(\gamma)$ for all $\gamma\in \widehat{G}$. Again, given a set $A\subset G$, define its normalized function as $\varphi_A = \frac{|G|}{|A|}1_A$.
We need two more definitions. Let $\Gamma \subset \widehat{G}$. Then $\Gamma^\perp$, called the \emph{annihilator} of $\Gamma$ is a subgroup of $G$ defined by $\Gamma^\perp = \{x\in G: \gamma(x) = 1, \forall \gamma \in \Gamma\}$. A subset $\Gamma \subset \widehat{G}$ is called \emph{dissociated} if there are no non-trivial solutions to the equation $$\sum_{\gamma \in \Gamma} a_\gamma\cdot \gamma= 0  $$ where each $a_\gamma\in \{-1,0,1\}$, $1\cdot \gamma = \gamma$,  $(-1)\cdot\gamma = -\gamma$, and $0\cdot\gamma= 0$. 
Let us restate the general form of Chang's lemma~\cite{chang2002polynomial}.
\begin{lemma}\label{lemma:genchang}
Let $A\subset G$ be a set of density $\alpha>0$.  Suppose that $\Gamma\subset \widehat{G}$ is a dissociated set. Then 
$$\sum_{\gamma\in \Gamma} |\widehat{\varphi_A}(\gamma)|^2\leq O(\log \alpha^{-1}).$$
\end{lemma}
Note that if $\Gamma\subset \widehat{G}$ and $\Gamma'\subset \Gamma$ is the largest dissociated subset of $\Gamma$, then $\Gamma \subset \langle \Gamma' \rangle$, the span of $\Gamma'$. Since $G$ has exponent $m$, we have $|\Gamma|\leq m^{|\Gamma'|}$. Moreover, 
one can show that $\Gamma^\perp \cong G/ \langle\Gamma\rangle$, therefore, $|\langle\Gamma^\perp\rangle|\geq \frac{|G|}{m^{|\Gamma'|}}$.

Finally, the last  part to modify in the proof of \Cref{theorem:main3,theorem:main4} is to obtain a suitable version of \Cref{claim:fourier}. Using Chang's lemma as stated above and following analogous proof as before, we can find the function $h':G\rightarrow \mathbb{C}$ taking values in the unit circle, and also the sets $A_1,\cdots,A_N \subset G$ as discussed in the proof of \Cref{theorem:main}. Also using \Cref{lemma:genchang} we can find a maximal dissociated subset $\Gamma'$ of size $k=O(c)$. Then take the subgroup  $H =\Gamma'^\perp\leq G$ (as the analog of the subspace $V\leq \ftwo^n$) so that $|G/H|\leq m^k$. The following claim about $H$ is what we need. 
\begin{claim}\label{claim:fouriergen}
 Let $\y_1 \in A_1,\ldots, \y_N \in A_N, \v \in H$ be chosen uniformly and independently. Then for every $x \in G$ it holds that
$$
\left| \E[h'(x+\y_1+\ldots+\y_N)] - \E[h'(x+\y_1+\ldots+\y_N+\v)] \right| \le 2^{-N/8}|G|.
$$
\end{claim}
The proof is analogous to the proof of \Cref{claim:fourier}. By taking $N\geq 10n\log m$ we can make sure that $2^{-N/8}|G|\leq 2^{-\Omega(N)}$.
This finishes the proofs of \Cref{theorem:main3,theorem:main4}.
\bibliographystyle{alpha}
\bibliography{linsketch}

\appendix
\section{Pseudorandom generators}\label{app:prg}
Below we describe a standard application of pseudorandom generators for space bounded computation by Nisan~\cite{N90} in the streaming setting following presentation given by Indyk~\cite{I06}.
Such pseudorandom generators can be used to fool any Finite State Machine which uses only $O(S)$ space or $2^{O(S)}$ states.
Since a sketch consisting of $s$ numbers modulo $p$ can only take $p^s$ values we can think of such sketches as being Finite State Machines with $O(s \log p)$ space.

Assume that a Finite State Machine $Q$ which uses $O(S)$ bits of space uses at most $k$ blocks of random bits where each block is of length $b$.
The generator $G \colon \{0,1\}^m \to (\{0,1\}^b)^k$ expands a small number $m$ of uniformly random bits into $kb$ bits which ``look random'' for $Q$.
Formally, let $U(\{0,1\}^t)$ be a uniform distribution over $\{0,1\}^t$.
For any discrete random variable let $D[X]$ be the distribution of $X$ interpreted is a vector of probabilities.
Let $Q(x)$ denote the state of $Q$ after using the random bits sequence $x$.
Then $G$ is a pseudorandom generator \textit{with parameter} $\epsilon > 0$ for a class $\mathcal C$ of Finite State Machines, if for every $Q \in \mathcal C$:
$$|D[Q_{x \sim U(\{0,1\}^{bk})}] - D[Q_{x \sim U(\{0,1\}^m)} (G(x))]|_1 \le \epsilon,$$
where $|y|_1$ denotes the $\ell_1$-norm of a vector $y$.

\begin{theorem}\cite{N90}
There exists a pseudorandom generator $G$ with parameter $\epsilon = 2^{-O(S)}$ for Finite State Machines using space $O(S)$ such that:
\begin{itemize}
\item $G$ expands $O(S \log R)$ bits into $O(R)$ bits.
\item $G$ requires only $O(S)$ bits of storage (in addition to its random input bits)
\item Any length-$O(S)$ block of $G(x)$ can be computed using $O(\log R)$ arithmetic operations on $O(S)$-bit words.
\end{itemize}
\end{theorem}

Using the above results we can reduce the amount of randomness used by a linear sketch modulo $p$ as follows.
Consider any state $\mathcal S$ of the linear sketch of dimension $s$.
From the above discussion it follows that this state can be represented using $O(s \log p)$ bits.
When a streaming update to coordinate $i$ arrives we need only $O(s \log p)$ bits to generate the $i$-th row of the linear sketch matrix so that we can add it to the linear sketch. However, in order to ensure consistency, i.e. to make sure that when the $i$-th row is generate again we get the same result, we still need a lot of memory. Solution to this issue due to~\cite{I06} follows below.

First, assume that the streaming updates $(i,\delta_t)$ are coming in the non-decreasing order of $i$.
In this case we don't have to store the rows of the linear sketch matrix as we can generate them on the fly.
Indeed, after the $i$-th row is generated we can apply it to all updates which contain $i$ since such updates arrive sequentially one after another.
This gives an algorithm which uses only $O(s \log p)$ storage and $O(n)$ blocks of random bits of size $O(s \log p)$ each.
Hence there exists a pseudorandom generator $G$ which given a random seed of size $O(s \log p \log (n/\delta))$ expands it to a pseudorandom sequence using which instead of the rows the sketch matrix only results in a negligible probability of error.
I.e. the resulting state of the sketch can still be used to estimate the value of the function $f$ of interest.

The key observation is that for every fixed random seed the resulting state doesn't depend on the order of updates $(i,\delta_t)$ by the commutativity of addition. Hence one can use $G$ even if the updates come in any order.

\end{document}